\documentclass[%
twocolumn,
amsmath,amssymb,
aps, physrev,
]{revtex4-2}

\usepackage{graphicx}% Include figure files
\usepackage{dcolumn}% Align table columns on decimal point
\usepackage{bm}% bold math
\usepackage{amsthm}
\usepackage{braket}
\newtheorem{theorem}{Theorem}

\begin{document}
	
	\preprint{APS/123-QED}
	
	\title{\textbf{Purity and bound energy in ancilla-assisted work extraction} 
	}% 

	\author{B. Vigneshwar$^\dagger$}
	\author{Farhaan Khan}
	\author{R. Sankaranarayanan}%
	\affiliation{Department of Physics, National Institute of Technology, Tiruchirappalli, 620015, TamilNadu, India.\\ $^\dagger$ vigneshwarbvb@gmail.com}%Lines break automatically or can be forced with \\

	\date{\today}% It is always \today, today,
	
	\begin{abstract}
		We investigate ancilla-assisted work extraction in quantum batteries from the perspective of bound energy and purity. We show that the bound energy of the reduced system provides a tight upper bound to the daemonic gain and that this bound is saturated for globally pure system--ancilla states. Motivated by this relation, we introduce a purity-based gain that qualitatively predicts the daemonic gain without requiring explicit optimization over measurements. We further introduce a protocol to analyze the role of dissipation and intrinsic interactions on daemonic gain. Under a collective environment, dissipation can dynamically generate and stabilize finite daemonic gain through environment-induced correlations. In interacting systems, level crossings and spectral restructuring strongly modify the attainable gain through their influence on the accessible bound energy. Our results demonstrate that daemonic gain is governed not only by correlations, but also by the spectral structure of the underlying Hamiltonian and information loss captured by bound energy and purity.
	\end{abstract}

	\maketitle
	
	\section{Introduction}
	\label{sec:introduction}
	
	The concept of ergotropy quantifies the maximum amount of
	work that can be extracted from a finite quantum system through cyclic
	unitary operations \cite{allahverdyan2004maximal}. It provides a
	rigorous distinction between passive states, from which no work can be
	extracted, and active states that contain usable energetic
	resources \cite{goold2016role,bera2019thermodynamics}. Ergotropy has thus become
	central to the study of work extraction
	protocols that harness entanglement and coherence in quantum batteries \cite{alicki2013entanglement,acin2018quantum,vinjanampathy2016quantum,campaioli2024colloquium,campaioli2017enhancing,francica2020quantum,rossini2019many,ahmadi2024nonreciprocal,song2024remote,lu2025topological}.
	
	The introduction of Maxwell's daemon within this paradigm gave rise to the idea of gathering information through measurements that make the system more active \cite{maruyama2009colloquium}. Such an enhancement in ergotropy, known as \textit{daemonic gain},
	originates from quantum correlations such as entanglement and
	discord, which enable conditional operations to unlock additional ergotropic contributions \cite{francica2017daemonic}. Recent studies have demonstrated that protocol design with ancilla-assisted networks display higher efficiency in charging power and extraction~\cite{bernards2019daemonic,vsafranek2023work,morrone2023daemonic,chaki2025auxiliary,satriani2024daemonic,kua2026daemonic}. 
	
	In addition to theoretical advancements, recent experiments have demonstrated the promise of enhanced quantum batteries \cite{quach2022superabsorption,joshi2022experimental,maillette2023experimental}. However, environmental noise directly reduces the amount of extractable work and overall efficiency \cite{hu2022optimal,razzoli2025cyclic,gemme2022ibm,ferraro2026opportunities}. This has motivated increasing attention toward understanding the role of dissipation in quantum thermodynamic protocols \cite{zhao2021quantum,barra2019dissipative,tirone2025quantum,tirone2023work,shastri2025dephasing,vigneshwar2026noise}. 
	
	At the same time, the extractable work in interacting quantum systems is strongly influenced by the underlying spectral structure. In particular, energy-level spacing and spectral transitions determine both the storage capacity and the amount of inaccessible energy present in the system \cite{yang2023battery}. States possessing similar correlations can therefore exhibit significantly different thermodynamic behavior depending on the associated energy gaps and level organization \cite{vigneshwar2025nonlocal}. 
	
	Daemonic gain requires an optimization over measurement protocols, which rapidly becomes computationally demanding for larger systems. Moreover, while correlations are known to enhance ancilla-assisted work extraction, the precise thermodynamic role of bound energy and purity remains unclear, particularly in open quantum systems. While dissipation can degrade correlations through irreversible information loss, it can also stabilize extractable work and dynamically generate useful correlations through collective environmental coupling~\cite{shastri2025dephasing,canzio2025single,wang2025global}. Furthermore, interaction-induced level crossings and gap closings can qualitatively alter the amount of bound energy stored in the system, potentially modifying the attainable daemonic advantage. Understanding how these mechanisms jointly influence ancilla-assisted work extraction therefore constitutes an important problem in realistic quantum batteries.
	
	Motivated by these considerations, in this work we reinterpret daemonic gain from a thermodynamic perspective as the conversion of bound energy into extractable work mediated by correlations between the system and the ancilla. Building on this interpretation, we show that the bound energy of the reduced system provides a tight upper bound to the daemonic gain and that this bound is saturated for globally pure system--ancilla states. We further introduce a purity-based predictor of daemonic gain that circumvents the explicit optimization over measurements while successfully capturing the qualitative behavior of the accessible thermodynamic advantage.
	
	We then establish a protocol to investigate ancilla-assisted work extraction under collective environment. We show that dissipation can dynamically generate and stabilize finite daemonic gain even in the absence of initial correlations through environment-induced correlation buildup. At the same time, irreversible information loss limits the efficiency with which bound energy is converted into extractable work, establishing the nontrivial role of dissipation in controlling the operational usefulness of correlations in quantum batteries.
	
	We then introduce intrinsic interactions in the system Hamiltonian and demonstrate that interaction-driven spectral restructuring strongly influences the attainable daemonic gain. In particular, level crossings and gap closings suppress the thermodynamic utility of correlations through the reduction of bound energy, while interaction-induced gap enhancement amplifies the ancilla-assisted advantage. Even in the presence of dissipation, spectral transitions manifest as pronounced signatures in the purity-based gain. Therefore, our work illuminates the role of purity and bound energy in ancilla-assisted protocols by illustrating the effects of information loss and spectral transitions.

	Our work is structured as follows. In Sec.~\ref{sec2}, we establish the relation between daemonic gain, bound energy, and purity. In Sec.~\ref{sec3}, we analyze ancilla-assisted work extraction in non-interacting open quantum batteries. Section~\ref{sec4} investigates the influence of spectral restructuring on the gain in interacting systems. Finally, the key results are summarized in Sec.~\ref{con}.

	\section{Daemonic Gain and Purity}\label{sec2}
	
	Ergotropy quantifies the maximum amount of work extractable from a quantum state through a cyclic unitary process \cite{allahverdyan2004maximal}. It is defined as
	\begin{equation}
		\xi(\rho, H) = E(\rho) - E_p(\rho) 
		= \operatorname{Tr}(\rho H) - \min_{V} \operatorname{Tr}(V \rho V^\dagger H),
		\label{e1}
	\end{equation}
	where $E(\rho)=\operatorname{Tr}(\rho H)$ denotes the average energy of the state $\rho$, and $E_p(\rho)$ is the energy of its passive counterpart, obtained by minimizing over all unitary operations.
	Let $\rho = \sum_n r_n \ket{r_n}\!\bra{r_n}$ and 
	$H = \sum_n e_n \ket{e_n}\!\bra{e_n}$ 
	be the spectral decompositions of the state and the Hamiltonian, respectively. The passive state relative to $H$ is given by
	\begin{equation}
		\rho^p = \sum_n r_n^\downarrow \ket{e_n^\uparrow}\!\bra{e_n^\uparrow},
	\end{equation}
	where $\{r_n^\downarrow\}$ and $\{e_n^\uparrow\}$ denote the eigenvalues of $\rho$ and $H$ arranged in decreasing and increasing order, respectively. Consequently, the passive energy takes the form
	\begin{equation}
		E_p(\rho) = \sum_n r_n^\downarrow e_n^\uparrow .
	\end{equation}
	
	The concept of ancilla-assisted work extraction was introduced in \cite{francica2017daemonic}, where information gained from measurements performed on an ancillary system is exploited to enhance work extraction. For a bipartite system--ancilla state $\rho_{SA}$, the \emph{ancilla-assisted ergotropy} is defined as
	\begin{equation}
		\xi_{\Pi}(\rho_{SA},H)
		= \operatorname{Tr}[\rho_S H]
		- \sum_a p_a \min_U \operatorname{Tr}[U \rho_{S|a} U^\dagger H],
		\label{eqep}
	\end{equation}
	where $\rho_{S|a}=\operatorname{Tr}_A[\pi_a^A \rho_{SA} \pi_a^A]/p_a$ is the conditional state of the system corresponding to measurement outcome $a$ on the ancilla, occurring with probability $p_a=\operatorname{Tr}[\pi_a^A \rho_{SA}]$ and $\rho_S=\operatorname{Tr}_A[\rho_{SA}]$ is the reduced system state. Here $\{\pi_a^A\}$ are rank-one projective measurements on the ancilla.
	The \emph{daemonic gain} is defined as the maximal enhancement in ergotropy due to the measurement protocol:
	\begin{equation}
		\delta W= \max_\Pi \xi_{\Pi}(\rho_{SA},H)-\xi(\rho_S,H),
		\label{eq1}
	\end{equation}
	where $\xi_{\Pi}$ denotes the ancilla-assisted ergotropy after projective measurement.
	
	Another relevant thermodynamic quantity is the \emph{bound energy}, defined as the difference between the passive energy and the ground state energy $\epsilon_0$ of the Hamiltonian \cite{bera2019thermodynamics}:
	\begin{equation}
		E_b=\min_{U}\operatorname{Tr}[U\rho_S U^\dagger H]- \epsilon_0.
		\label{eq2}
	\end{equation}
	Bound energy quantifies the portion of energy that cannot be extracted as work via unitary operations. Notably, pure states possess zero bound energy, since an appropriate unitary can always transform them to the ground state. Hence, bound energy is intrinsically connected to the level of mixedness or impurity of the quantum state \cite{mula2023ergotropy}.
	For a composite state $\rho_{SA}$, the purity of the reduced system state $\rho_S$ reflects the degree of bound energy in the system. Daemonic gain arises from correlation between system and ancilla, which also affects the purity of the reduced states. Hence, it is natural to anticipate a relationship between $\delta W$ and $E_b$.
	
	From this perspective, ancilla-assisted work extraction may be viewed as a process that converts bound energy into free energy at the expense of correlations. Consequently, $E_b$ constitutes a natural upper bound for the daemonic gain. 
	This can be shown directly from Eqs.~\eqref{eq1} and \eqref{eq2} as
	\begin{equation}
		E_b - \delta W
		=
		\min_\Pi \sum_a p_a \min_U \operatorname{Tr}[U \rho_{S|a} U^\dagger H]
		- \epsilon_0 .
	\end{equation}
	Since $\epsilon_0$ is the lowest eigenvalue of $H$, one has
	\begin{equation}
		\min_U \operatorname{Tr}[U \rho_{S|a} U^\dagger H] \ge \epsilon_0
	\end{equation}
	for every conditional state $\rho_{S|a}$. Therefore,
	\begin{equation}
		E_b - \delta W \ge 0,
	\end{equation}
	which establishes $\delta W \le E_b$. Whether this bound is saturated depends on the structure of the initial system--ancilla state. The following theorem establishes a direct connection between bound saturation and global purity.
	
	\begin{figure*}[ht]
		\centering
		\includegraphics[width=1\linewidth]{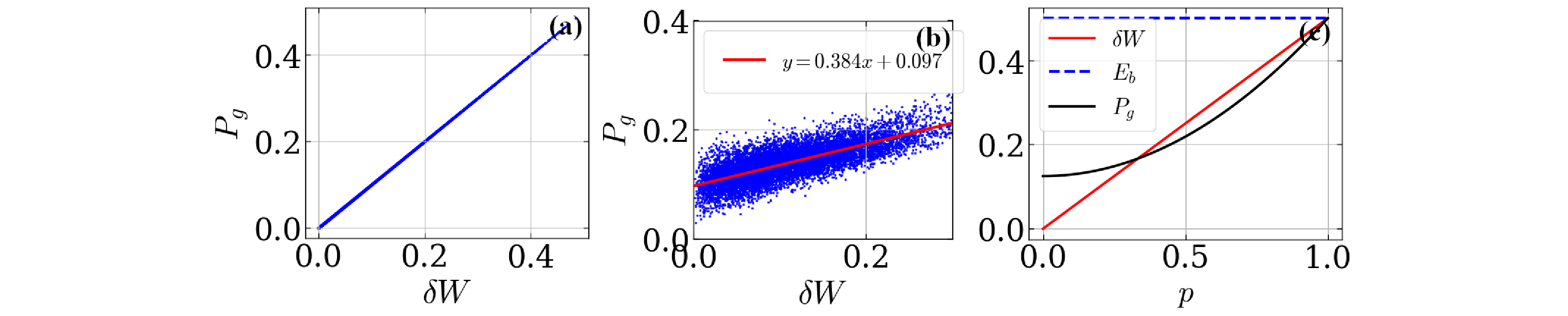}
		\caption{Scatter plot of $P_g$ and $\delta W$ for $10^4$ randomly generated (a) pure states and (b) mixed states. (c) $P_g$, $E_b$ and $\delta W$ as a function of parameter $p$ for Werner states.}
		\label{f1}
	\end{figure*}
	
	\begin{theorem}
		The upper bound of the daemonic gain, $E_b$, is saturated whenever the system--ancilla state $\rho_{SA}$ is pure, i.e., $\rho_{SA}=\ket{\psi_{SA}}\bra{\psi_{SA}}$.
	\end{theorem}
	
	\begin{proof}
		It suffices to show that if $\rho_{SA}$ is pure, then the conditional state $\rho_{S|a}$ is also pure for every projective measurement outcome $a$. In this case, the minimization in Eq.~\eqref{eqep} allows the passive state to reach the ground state, implying vanishing bound energy.
		
		Let the bipartite pure state admit the Schmidt decomposition
		\begin{equation}
			\ket{\psi_{SA}}=\sum_{i} \sqrt{\lambda_i}\,
			\ket{\alpha_i}\ket{\beta_i},
		\end{equation}
		with $\sum_i \lambda_i=1$. Acting with the projector $\mathbb{I}\otimes \pi_a^A$ and tracing out the ancilla yields
		\begin{equation}
			\rho_{S|a}
			=\frac{1}{p_a}
			\sum_{ij}\sqrt{\lambda_i\lambda_j}
			\bra{\beta_i}\pi_a^A\ket{\beta_j}\ket{\alpha_j}\bra{\alpha_i},
		\end{equation}
		where $p_a=\sum_i \lambda_i \bra{\beta_i}\pi_a^A\ket{\beta_i}$.
		
		Taking $\pi_a^A=\ket{a}\bra{a}$ and defining 
		$\bra{\beta_i}\pi_a^A\ket{\beta_j}
		=\langle\beta_i|a \rangle\langle a|\beta_j\rangle
		=a_i^* a_j$,
		the purity of the conditional state becomes
		\begin{equation}
			\operatorname{Tr}[\rho_{S|a}^2]
			=\frac{\sum_{ij}\lambda_i\lambda_j |a_i|^2 |a_j|^2}{p_a^2}
			=\frac{p_a^2}{p_a^2}=1,
		\end{equation}
		which confirms that $\rho_{S|a}$ is pure.
	\end{proof}
	
	The above result demonstrates that the purity of the global state plays a central role in daemonic work extraction. In general, the daemonic gain depends on three key factors: the bound energy stored in the reduced system state, the correlations shared between the system and ancilla, and the purity of the global system--ancilla state. While bound energy quantifies the amount of inaccessible energy stored in the reduced state, the correlations determine how efficiently this energy can be converted into extractable work through measurements on the ancilla.
	
	Motivated by these observations, we introduce the quantity \emph{purity-based gain} defined as
	\begin{equation}
		P_g=\operatorname{Tr}[\rho_{SA}^2] E_b,
		\label{eq3}
	\end{equation}
	as a qualitative predictor of the daemonic gain. The motivation behind this construction is twofold. First, the explicit optimization over measurement protocols in Eq.~\eqref{eq1} becomes computationally demanding for larger systems. Second, the theorem above indicates that global purity directly controls the efficiency with which bound energy can be converted into ergotropy. Consequently, $P_g$ combines the available bound energy with the degree of global purity to estimate the accessible daemonic advantage without requiring measurement optimization.
	
	It is evident from Eq.~\eqref{eq3} that $P_g$ is invariant under local unitary operations and takes values in the range $(0,E_b)$. To illustrate the relationship between $\delta W$ and $P_g$, we generate $10^4$ random two-qubit states, identifying one qubit as the system and the other as the ancilla. The system Hamiltonian is chosen as $H=\sigma_z/2$. 
	For the evaluation of $\delta W$, projective measurements of the form
	\begin{align}
		P_0 &= \cos k \ket{0} + \sin k \ket{1}, \nonumber \\
		P_1 &= -\sin k \ket{0} + \cos k \ket{1},
		\label{peq}
	\end{align}
	are considered, with optimization over $k\in(0,\pi/2)$. In contrast, the evaluation of $P_g$ is straightforward and does not require measurement optimization.
	
	For pure system--ancilla states, the scatter plot of $P_g$ and $\delta W$ collapses to a straight line, as illustrated in Fig.~\ref{f1}(a). In this case, $P_g$ becomes identical to the bound energy $E_b$. Since globally pure states saturate the upper bound established in Theorem~1, one obtains $\delta W=E_b$, leading to an exact correspondence between $P_g$ and $\delta W$. 
	For mixed states, shown in Fig.~\ref{f1}(b), the relationship exhibits a broader spread. Here the global purity is reduced and the upper bound is no longer saturated, resulting in quantitative deviations between $P_g$ and $\delta W$. Nevertheless, the overall trend of the daemonic gain is still qualitatively captured by $P_g$, as evidenced by the fitted red line in Fig.~\ref{f1}(b). The spread reflects the fact that, for mixed states, correlations and bound energy are not converted into extractable work with unit efficiency.
	
	The role of purity becomes even more transparent by considering the Werner state
	\begin{equation}
		\rho(p)
		=
		p\,|\Phi^{+}\rangle\!\langle\Phi^{+}|
		+
		(1-p)\frac{\mathbb{I}_4}{4},
	\end{equation}
	where
	\begin{equation}
		|\Phi^{+}\rangle
		=
		\frac{1}{\sqrt{2}}
		\left(
		|00\rangle+|11\rangle
		\right),
	\end{equation}
	and $\mathbb{I}_4$ denotes the identity operator in the two-qubit Hilbert space. The reduced state of the system remains maximally mixed for all values of $p$, implying a constant bound energy $E_b=0.5$ for $H=\sigma_z/2$. However, the daemonic gain increases monotonically with $p$, as shown in Fig.~\ref{f1}(c). At $p=0$, the state is maximally mixed and uncorrelated, yielding vanishing daemonic gain. Increasing $p$ enhances the contribution of the Bell state component, thereby amplifying the accessible correlations and increasing the gain. Although the bound energy remains constant throughout, the qualitative behavior of $\delta W$ is accurately reflected by $P_g$. This demonstrates that bound energy alone is insufficient to characterize the daemonic advantage and that the global purity plays an essential operational role.
	
	Taken together, these observations clarify the operational significance of the purity-based gain. For globally pure system--ancilla states, $P_g$ becomes exact because the bound-energy upper limit is saturated. For mixed states, $P_g$ remains a reliable qualitative indicator of the daemonic gain, capturing how efficiently correlations convert bound energy into extractable work. 
	Physically, the bound energy of the reduced system originates from information loss induced by correlations with external degrees of freedom. Measurements performed on the ancilla can partially recover this information and thereby convert a fraction of the bound energy into ergotropy. However, when the combined system--ancilla state itself becomes mixed, additional information is irreversibly lost to inaccessible external degrees of freedom, such as the environment. This inaccessible loss suppresses the attainable daemonic gain even in the presence of finite bound energy. The quantity $P_g$ captures this mechanism by incorporating the global purity of the system--ancilla state, thereby providing a compact operational estimate of the accessible thermodynamic advantage.
	
	\begin{figure}[ht]
		\centering
		\includegraphics[width=1\linewidth]{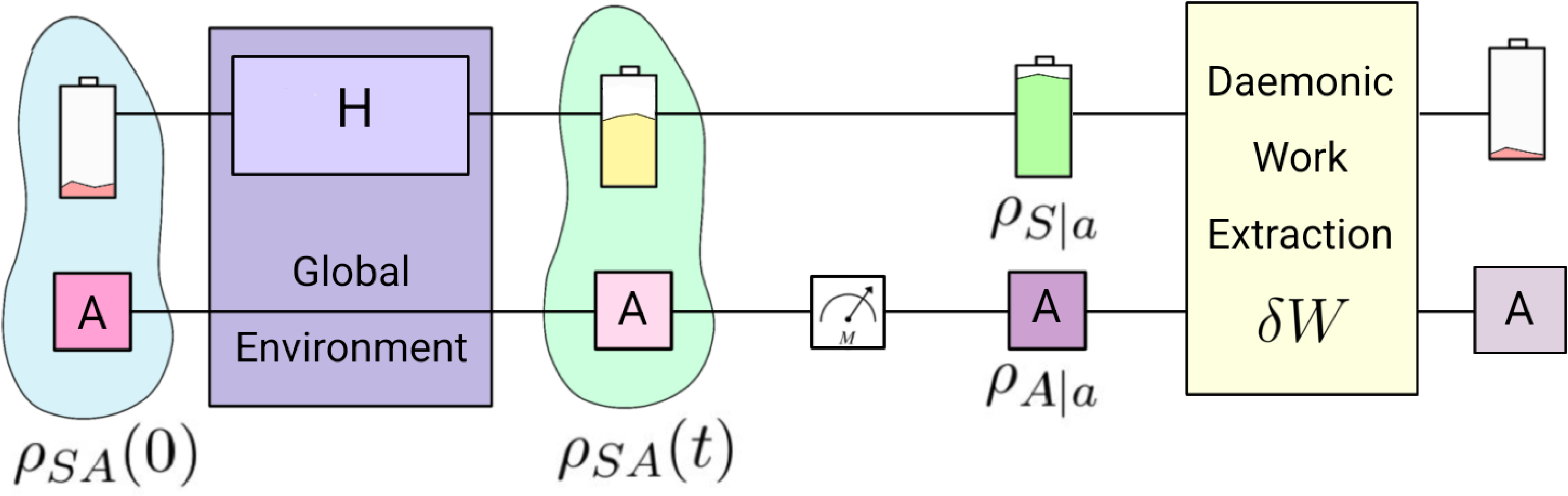}
		\caption{Schematic diagram of the protocol. The system and ancilla are initially correlated with the system occupying lower energies followed by charging under global environment. The ancilla is then measured to implement daemonic work extraction.}
		\label{f5}
	\end{figure}

	\section{Ancilla-assisted open system}\label{sec3}
	
	In the previous section, we established that daemonic gain is intrinsically connected to both the bound energy of the reduced system and the purity of the global system--ancilla state. Since purity is highly susceptible to environmental noise, it is essential to investigate how dissipation influences ancilla-assisted work extraction. In realistic quantum batteries, unavoidable coupling to the environment induces entropy production, suppresses coherence, and modifies system--ancilla correlations. Understanding the interplay between noise, correlations, and bound energy is therefore crucial for assessing the operational viability of daemonic protocols.
	
	\subsection{Protocol}
	To analyze environmental effects in ancilla-assisted protocol, we consider a minimal setting in which the battery is modeled by the system Hamiltonian $H_S$, while the ancilla acts solely as a resource for measurement-assisted work extraction. The charging field $H_C$ acts directly and exclusively on the system via a standard direct charging protocol \cite{campaioli2024colloquium,ghosh2020enhancement,le2018spin}, thereby eliminating coherent system--ancilla interactions during the charging stage. Although the ancilla may alternatively be viewed as an auxiliary charger \cite{satriani2024daemonic}, where cyclic unitary charging can close the daemonic gap~\cite{pushpan2025DaemonicGap}, such gap closing is generally inhibited in the presence of dissipation due to the generation of bound energy by environmental coupling. The present protocol therefore enables a transparent identification of the role of dissipation in ancilla-assisted work extraction.
	
	\begin{figure*}[ht]
		\centering
		\includegraphics[width=1\linewidth]{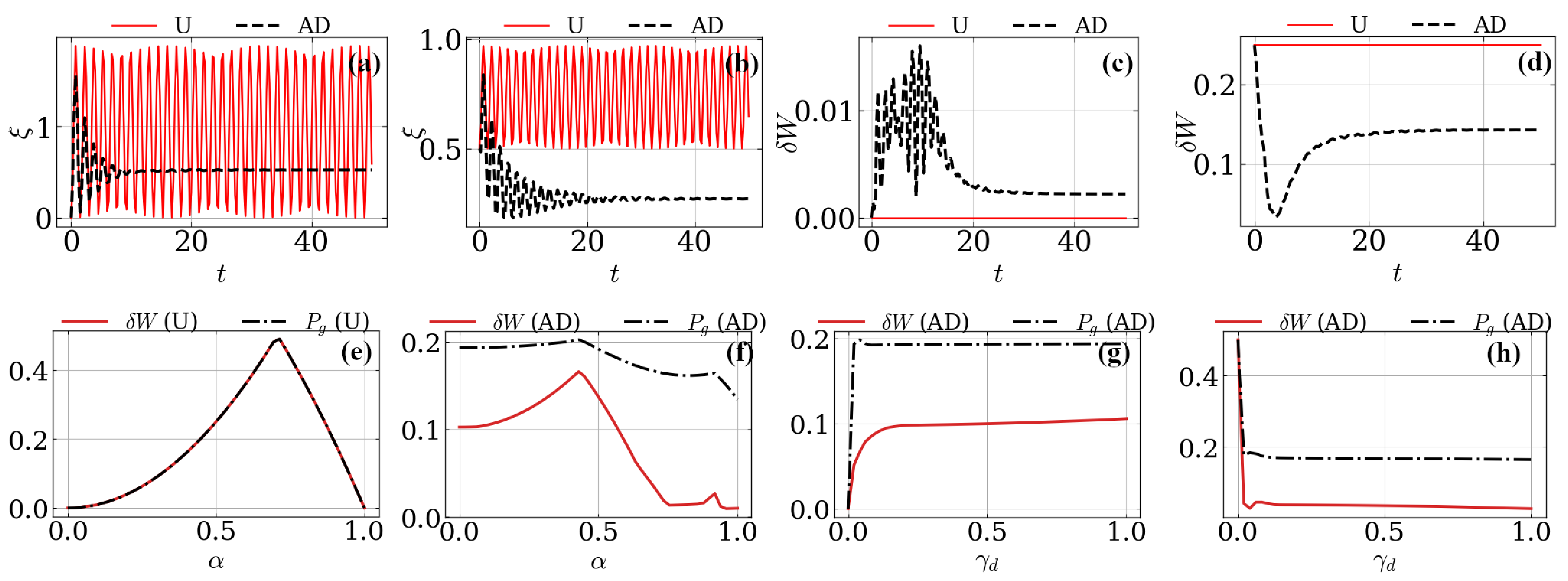}
		\caption{Ergotropy of the two-qubit battery under unitary evolution and collective amplitude damping for (a) $\alpha=1$ and (b) $\alpha=0.5$. The corresponding daemonic gain is shown for (c) $\alpha=1$ and (d) $\alpha=0.5$. Comparison between the daemonic gain $\delta W$ and the purity-based gain $P_g$ at $t=100$ as a function of the initial-state parameter $\alpha$ for (e) unitary evolution and (f) amplitude damping for $\gamma_d=0.3$. Gain as a function the dissipation strength $\gamma_d$ for (g) $\alpha=0$ and (h) $\alpha=1/\sqrt{2}$, respectively. Other parameters are $\omega_c=2$ and $h=1$.}
		\label{f2}
	\end{figure*}
	
	We initialize the system--ancilla state in the correlated form
	\begin{equation}
		\ket{\psi_{SA}(0)}
		=
		\alpha \ket{\psi_0}\ket{0}
		+
		\beta \ket{\psi_1}\ket{1},
	\end{equation}
	where $\ket{\psi_0}$ and $\ket{\psi_1}$ denote the ground and first excited eigenstates of the system Hamiltonian $H_S$, $\beta=\sqrt{1-\alpha^2}$, and $\rho_{SA}(0)=\ket{\psi_{SA}(0)}\bra{\psi_{SA}(0)}$. The parameter $\alpha$ controls both the initial population imbalance and the degree of system--ancilla correlation.
	
	The subsequent evolution is governed by the Lindblad master equation
	\begin{equation}
		\dot{\rho}_{SA}(t)
		=
		-i[H_S+H_C,\rho_{SA}(t)]
		+
		\gamma_d D[O](\rho_{SA}(t)),
		\label{eq5}
	\end{equation}
	where $H_C$ denotes the charging Hamiltonian and 
	\begin{equation}
		D[O](\rho)
		=
		O\rho O^\dagger
		-\frac{1}{2}(O^\dagger O \rho+\rho O^\dagger O)
	\end{equation}
	describes the dissipative contribution with rate $\gamma_d$. After the evolution, ergotropy and daemonic gain are computed with respect to the system Hamiltonian $H_S$. The schematic diagram of the protocol is shown in Fig.~\ref{f5}.
	
	We first examine two non-interacting system qubits coupled to a single ancilla qubit. The system and charging Hamiltonians are given by
	\begin{equation}
		H_S=\frac{h}{2}(\sigma_z^{(1)}+\sigma_z^{(2)}),
		\qquad
		H_C=\omega_c(\sigma_x^{(1)}+\sigma_x^{(2)}).
	\end{equation}
	For the dissipation, consider the collective amplitude damping (AD), modeled through the global Lindblad operator
	\begin{equation}
		O_{ad}=\sum_{i=1}^{3}\sigma_-^{(i)},
		\label{eqlb}
	\end{equation}
	where $\sigma_-^{(i)}=(\sigma_x^{(i)}-i\sigma_y^{(i)})/2$, and $i=1,2$ label the system qubits while $i=3$ labels the ancilla. The unitary case corresponds to $\gamma_d=0$.
	Such collective dissipators naturally arise when multiple qubits couple to a common reservoir, as in cavity QED systems, superconducting circuits, or spin ensembles interacting with a shared bosonic bath. Under the Born--Markov approximation, a common-bath interaction of the form $H_{\mathrm{int}} = \sum_i s_i \otimes B$ leads to effective Lindblad operators proportional to $\sum_i s_i$~\cite{wang2015dissipation,gelhausen2018dissipative}. Importantly, the resulting dissipator contains cross terms of the form $s_i \rho s_j$ ($i\neq j$), which mediate environment-induced correlations between otherwise non-interacting subsystems~\cite{wang2025global}. Consequently, collective dissipation provides a natural mechanism for studying noise-assisted daemonic gain.

	For unitary dynamics ($\gamma_d=0$), analytical expressions for the ergotropy and daemonic gain are provided in Appendix~\ref{appa}. The dissipative dynamics are obtained numerically by solving Eq.~\eqref{eq5} using the \textsc{QuTiP} package~\cite{johansson2012qutip}. In addition, measurement optimization introduced in \eqref{peq} is performed to calculate $\delta W$ numerically. The resulting ergotropy, daemonic gain, and purity-based gain are shown in Fig.~\ref{f2}.
	For $\alpha=1$, corresponding to an initially uncorrelated state, the ergotropy under unitary evolution exhibits persistent oscillations [see Fig.~\ref{f2}(a)], reflecting coherent energy exchange with the charging field. The oscillation frequency is determined by the interplay between the charging field strength $\omega_c$ and the intrinsic level spacing $h$, while the maximum attainable ergotropy depends additionally on the initial correlations through $\alpha$ [see Appendix~\ref{appa}] for the unitary evolution. In the presence of collective amplitude damping ($\gamma_d=0.3$), the oscillations are rapidly suppressed and the ergotropy approaches a stable non-zero value. Although dissipation reduces the maximal extractable work, it simultaneously stabilizes the output, consistent with earlier observations of noise-assisted stabilization in quantum batteries \cite{farina2019charger,gherardini2020stabilizing}.
	
	For $\alpha=0.5$, the initial correlations induce non-zero ergotropy already at $t=0$, as shown in Fig.~\ref{f2}(b). Under unitary dynamics, the ergotropy never vanishes during the evolution, indicating that initial system--ancilla correlations preserve a finite amount of extractable energy throughout the charging cycle. Under dissipative evolution, the oscillations are again damped and the system relaxes toward a steady ergotropic output lower than the initial value due to environmental decay.
	The corresponding daemonic gain is shown in Figs.~\ref{f2}(c) and \ref{f2}(d). For $\alpha=1$, the unitary evolution produces vanishing gain because the initial state is uncorrelated. In contrast, collective amplitude damping generates a finite daemonic gain despite the absence of initial correlations, as illustrated in Fig.~\ref{f2}(c). The gain initially exhibits temporal fluctuations before converging to a stable non-zero value. This behavior demonstrates that the common environment dynamically generates effective system--ancilla correlations, thereby activating noise-induced daemonic advantage.
	
	For initially correlated states ($\alpha=0.5$), the unitary daemonic gain remains constant throughout the evolution [see Fig.~\ref{f2}(d)], reflecting the preservation of correlations under coherent dynamics. Under collective dissipation, the gain becomes time dependent due to the competition between correlation generation and irreversible information loss to the environment. Nevertheless, unlike local dissipation, the collective channel preserves a finite amount of system--ancilla correlation and consequently stabilizes the daemonic gain at long times.
	
	Figure~\ref{f2}(e) shows the dependence of $\delta W$ and $P_g$ on the initial-state parameter $\alpha$ for unitary evolution. The gain vanishes at $\alpha=0$ and $\alpha=1$, corresponding to product initial states, and reaches its maximum at $\alpha=1/\sqrt{2}$, where the initial entanglement is maximal. Since globally pure states convert bound energy into ancilla-assisted work with unit efficiency, the purity-based gain coincides exactly with the daemonic gain in the unitary case.
	Under collective amplitude damping, shown in Fig.~\ref{f2}(f), the maximal gain shifts away from the maximally entangled configuration toward states with larger excited-state population ($\alpha<0.5$). In this regime, dissipative relaxation redistributes population while the collective jump operator simultaneously generates and preserves correlations, leading to enhanced and stabilized daemonic gain. Conversely, when the ground-state population dominates ($\alpha\approx1$), the environment suppresses correlations more rapidly than bound energy can be converted into extractable work, resulting in diminished gain.
	
	The qualitative agreement between $\delta W$ and $P_g$ persists even in the presence of noise. As shown in Fig.~\ref{f2}(g), both quantities increase from zero when the dissipation strength $\gamma_d$ is switched on for the initially uncorrelated state ($\alpha=0$), signaling environment-induced correlation generation. In contrast, for the maximally correlated initial state, increasing $\gamma_d$ suppresses the gain due to irreversible leakage of system--ancilla information into the larger environment  [see Fig.~\ref{f2}(h)]. Although collective dissipation mediates interactions between the subsystems, the attainable gain remains strongly constrained by the balance between initial correlations and dissipative information loss.
	
	Overall, these results demonstrate that the daemonic gain in open quantum systems is governed by a nontrivial interplay between dissipation, correlations, and bound energy. Collective dissipation does not merely degrade the thermodynamic advantage; rather, it can actively generate and stabilize useful correlations that convert bound energy into extractable work. At the same time, irreversible information loss to the environment limits the efficiency of this conversion process. The close qualitative agreement between $\delta W$ and $P_g$ across both coherent and dissipative regimes further establishes the purity-based gain as an effective indicator of the accessible daemonic advantage without requiring explicit measurement optimization.

	\section{Role of level spacing}\label{sec4}
	
	In the previous section, we examined how environmental dissipation and initial-state correlations influence the daemonic gain in non-interacting quantum batteries. Since the bound energy depends explicitly on the underlying spectrum, intrinsic interactions are expected to qualitatively modify the attainable daemonic advantage through changes in the energy-level structure. In interacting systems, tuning the coupling parameters can induce level crossings and spectral rearrangements, thereby altering the bound energy. In particular, degeneracies at level crossings can suppress the available bound energy, while interaction-induced gap enhancement can increase the thermodynamic utility of correlations.
	
	To investigate these effects, we consider an interacting two-qubit system described by an anisotropic Heisenberg $XYZ$ model supplemented with Dzyaloshinsky--Moriya interaction (DMI). Originating from spin--orbit coupling in systems lacking inversion symmetry \cite{dzyaloshinsky1958thermodynamic,moriya1960new}, the DMI introduces an antisymmetric exchange interaction that modifies both the eigenvalue spectrum and the structure of eigenstates \cite{vigneshwar2026noise}. Such interaction-induced spectral restructuring provides a natural platform to analyze how level spacing and level crossings influence ancilla-assisted work extraction.
	
	The system Hamiltonian is given by
	\begin{align}
		\hat{H}_S =\, & \frac{h}{2} \left( \hat{\sigma}_z^{(1)} + \hat{\sigma}_z^{(2)} \right) \nonumber \\
		& +\frac{J}{2} \left[ (1 + \gamma) \hat{\sigma}_x^{(1)} \hat{\sigma}_x^{(2)} 
		+ (1 - \gamma) \hat{\sigma}_y^{(1)} \hat{\sigma}_y^{(2)} \right] \nonumber \\
		& + \frac{J_z}{2} \hat{\sigma}_z^{(1)} \hat{\sigma}_z^{(2)} 
		+ \frac{D}{2} \left( \hat{\sigma}_x^{(1)} \hat{\sigma}_y^{(2)} 
		- \hat{\sigma}_y^{(1)} \hat{\sigma}_x^{(2)} \right),
		\label{2e1}
	\end{align}
	where $\hat{\sigma}_i^{(k)}$ ($i=x,y,z$) denotes the Pauli operator acting on the $k$th qubit. The parameters $J$ and $J_z$ quantify symmetric exchange interactions, $\gamma$ controls the anisotropy in the $xy$ plane, and $D$ characterizes the strength of the antisymmetric Dzyaloshinsky--Moriya interaction (DMI), which induces spin canting and enhances chiral correlations.
	
	The eigenstructure of $\hat{H}_S$ plays a central role in determining both ergotropy and bound energy. Diagonalizing Eq.~\eqref{2e1}, the eigenvalues and eigenvectors are obtained as
	\begin{subequations}
		\begin{align}
			e_{1,2} &= \pm \sqrt{J^2 + D^2} - \frac{J_z}{2}, \\
			|e_{1,2}\rangle &= 
			\frac{1}{\sqrt{|c_{1,2}|^2 + 1}}
			\left( 0,\, c_{1,2},\, 1,\, 0 \right)^\intercal, \\
			e_{3,4} &= \pm \sqrt{h^2 + J^2 \gamma^2} + \frac{J_z}{2}, \\
			|e_{3,4}\rangle &= 
			\frac{1}{\sqrt{|c_{3,4}|^2 + 1}}
			\left( c_{3,4},\, 0,\, 0,\, 1 \right)^\intercal,
		\end{align}
		\label{11e}
	\end{subequations}
	with coefficients
	\begin{align}
		c_{1,2} &= \pm \frac{J + i D}{\sqrt{J^2 + D^2}},
		\qquad
		c_{3,4} = \frac{h \pm \sqrt{h^2 + J^2 \gamma^2}}{J \gamma}.
	\end{align}
	
	\begin{figure}[ht]
		\centering
		\includegraphics[width=1\linewidth]{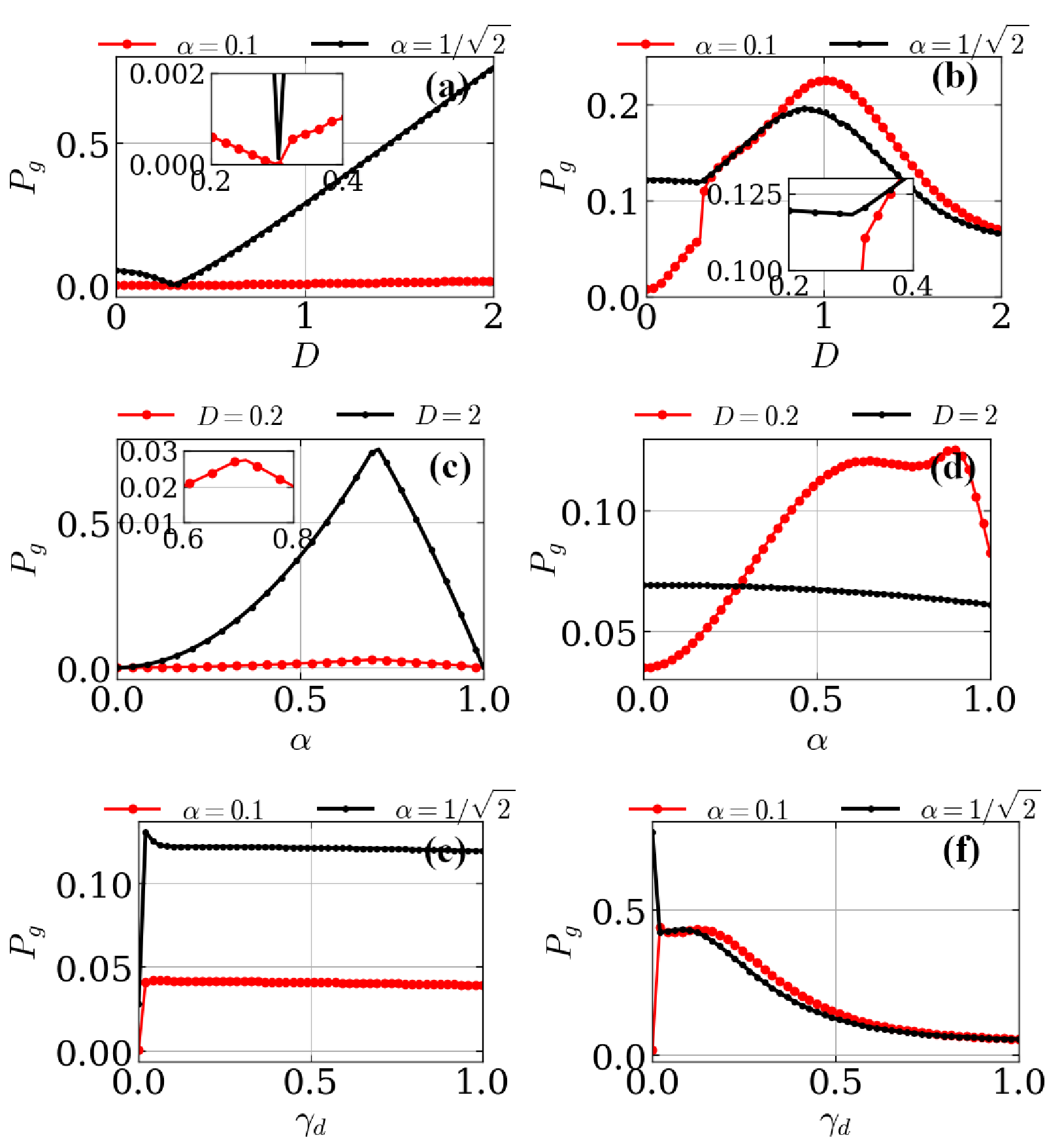}
		\caption{Purity-based gain $P_g$ for the interacting $XYZ$ battery model. $P_g$ as a function of DMI strength $D$ under (a) unitary evolution and (b) collective amplitude damping ($\gamma_d=0.8$), respectively, for $\alpha=0.1$ and $\alpha=1/\sqrt{2}$. Insets highlight the behavior near the critical point $D_c\approx0.3$. Dependence of $P_g$ on the initial-state parameter $\alpha$ for weak ($D=0.2$) and strong ($D=2$) DMI regimes under (c) unitary and (d) dissipative evolution, respectively. The inset illustrates the dependence of alpha for weak DMI near maximal entanglement. Dependence of $P_g$ on the dissipation strength $\gamma_d$ for (e) $D=0.2$ and (f) $D=2$. Other parameters are $h=1$, $J=0.4$, $\gamma=0.2$, $J_z=0.5$, and $\omega_c=5$.}
		\label{f4}
	\end{figure}
	
	Throughout this section, we focus on the ferromagnetic regime along the $z$ axis ($J_z>0$), where the competition between symmetric exchange and DMI determines the ordering of energy levels. Comparing the energies $e_2$ and $e_4$ yields the condition
	\begin{equation}
		\sqrt{D^2 + J^2}
		>
		\sqrt{h^2 + J^2 \gamma^2} - J_z,
		\label{12e}
	\end{equation}
	which defines a transition in the ground-state structure.
	
	The corresponding critical DMI strength, denoted $D_c$ \cite{vigneshwar2026noise}, separates distinct spectral regimes:
	\begin{itemize}
		\item $D<D_c$: symmetric exchange-dominated regime,
		\item $D=D_c$: ground-state degeneracy,
		\item $D>D_c$: DMI-dominated regime.
	\end{itemize}
	
	We now analyze the purity-based gain $P_g$ using the same charging protocol introduced in Sec.~\ref{sec3}, with the initial state
	\begin{equation}
		\ket{\psi_{SA}(0)}
		=
		\alpha \ket{\psi_0}\ket{0}
		+
		\beta \ket{\psi_1}\ket{1},
	\end{equation}
	where $\ket{\psi_0}$ and $\ket{\psi_1}$ now denote the ground and first excited eigenstates of the interacting Hamiltonian in Eq.~\eqref{2e1}. A relatively strong charging field with $\omega_c=5$ is employed to ensure that the enlarged spectral gap induced by DMI can be effectively accessed during the dynamics. The system evolves under the master equation up to $t=100$, after which the purity-based gain is evaluated.
	
	To investigate the role of spectral restructuring, we fix the interaction parameters as $h=1$, $J=0.4$, $\gamma=0.2$, and $J_z=0.5$, yielding a critical DMI strength of $D_c\approx0.3$. For the unitary case, the gain can be obtained analytically and depends explicitly on both the initial correlations and the energy spacing between the ground and first excited states [see Appendix~\ref{appa}]. Since the identities of $\ket{\psi_0}$ and $\ket{\psi_1}$ change across the transition, the daemonic gain acquires a direct sensitivity to spectral rearrangements absent in the non-interacting system.
	
	This behavior is illustrated in Fig.~\ref{f4}(a). Increasing the DMI strength initially suppresses the gain until the critical point $D_c$, where the lower energy levels become degenerate. At this point, the thermodynamic utility of the correlations is strongly reduced despite the presence of maximal entanglement, as highlighted in the inset of Fig.~\ref{f4}(a). Beyond the transition, the reopening and subsequent enhancement of the energy gap in the DMI-dominated regime increase the available bound energy, leading to a recovery and growth of $P_g$. Thus, the gain is governed not only by the magnitude of correlations but also by how efficiently the spectral structure supports the storage of bound energy.
	
	Under collective amplitude damping, shown in Fig.~\ref{f4}(b), the gain no longer vanishes completely at the transition point. Although the lower levels become degenerate near $D_c$, dissipation redistributes population across higher excited levels and dynamically generates finite bound energy. Nevertheless, the gain exhibits a pronounced nonanalytic variation across the transition due to the abrupt restructuring of the spectrum. The inset in Fig.~\ref{f4}(b) further highlights this sharp change in behavior near the critical point.
	
	The influence of the initial-state parameter $\alpha$ is shown in Figs.~\ref{f4}(c) and \ref{f4}(d). Under unitary evolution, the gain reaches its maximum near $\alpha=1/\sqrt{2}$ irrespective of the DMI strength, although the overall magnitude depends strongly on the spectral gap [see Fig.~\ref{f4}(c)]. In contrast, under dissipative evolution, the dependence on $\alpha$ becomes highly regime dependent. In the weak-DMI regime ($D=0.2$), intermediate values of $\alpha$ maximize the gain, indicating that both population imbalance and initial correlations significantly influence the dissipative conversion of bound energy into work. However, in the strong-DMI regime ($D=2$), the gain becomes nearly insensitive to $\alpha$, as shown in Fig.~\ref{f4}(d). In this regime, the enlarged spectral gap dominates the thermodynamic behavior, reducing the relative importance of the initial-state configuration.
	
	A similar distinction emerges in the dependence on the dissipation strength $\gamma_d$. In the weak-DMI regime, $P_g$ remains nearly constant with increasing $\gamma_d$ [see Fig.~\ref{f4}(e)], indicating that moderate dissipation does not substantially alter the accessible bound energy when the spectral gap is small. By contrast, in the DMI-dominated regime, increasing $\gamma_d$ suppresses the gain [see Fig.~\ref{f4}(f)] due to enhanced irreversible information loss associated with the larger DMI-dependent level spacing. Notably, the gain curves for $\alpha=0.1$ and $\alpha=1/\sqrt{2}$ nearly coincide at large $D$, reinforcing that the dominant contribution to the gain in this regime arises from the spectral structure rather than the initial correlations.
	
	These results demonstrate that daemonic gain in interacting quantum batteries is fundamentally shaped by the underlying spectral organization of the system. Level crossings and gap closings suppress the thermodynamic utility of correlations by reducing the accessible bound energy, whereas interaction-induced gap enhancement can amplify the ancilla-assisted advantage. Importantly, the gain exhibits clear signatures of spectral transitions even in the presence of dissipation, indicating that ancilla-assisted work extraction is sensitive not only to correlations but also to the restructuring of the Hamiltonian spectrum itself. These observations establish the purity-based gain as an effective probe of interaction-driven spectral phenomena and provide a direct thermodynamic interpretation of how level spacing controls the operational usefulness of correlations in ancilla-assisted quantum batteries.

	\section{Conclusion}\label{con}
	
	In this work, we investigated ancilla-assisted work extraction in quantum batteries from the perspective of bound energy and purity. We showed that the bound energy of the reduced system provides a natural upper bound to the daemonic gain and that this bound is saturated for globally pure system--ancilla states. This establishes a direct thermodynamic interpretation of daemonic protocols, where correlations enable the conversion of otherwise inaccessible bound energy into extractable work.
	Motivated by this connection, we introduced the purity-based quantity $P_g$ as a computationally efficient indicator of daemonic gain. For globally pure states, $P_g$ reproduces the exact gain, while for mixed states it successfully captures the qualitative behavior without requiring explicit optimization over measurements. The close agreement between $\delta W$ and $P_g$ demonstrates that purity plays a central role in determining the operational usefulness of correlations for work extraction.
	
	For non-interacting quantum batteries, we analyzed the influence of collective dissipation and showed that noise can both suppress and generate daemonic advantage through environment-induced correlations, highlighting that dissipation can actively participate in the conversion of bound energy into useful work rather than merely degrading performance.
	We then demonstrated that intrinsic interactions qualitatively reshape the daemonic gain through interaction-driven modifications of the energy spectrum. While ground state degeneracies suppress the accessible bound energy and reduce the thermodynamic utility of correlations, interaction induced gap enhancement amplifies the attainable ancilla-assisted advantage. These effects persist even in the presence of dissipation, indicating that daemonic gain is sensitive to the restructuring of the underlying spectrum.
	
	In summary, our results establish that daemonic gain is governed not only by correlations, but also by the spectral organization of the system through its influence on bound energy. Since the purity-based gain $P_g$ responds sensitively to level crossings and spectral transitions, the present framework can be naturally extended to larger many-body systems to investigate the behavior of daemonic gain near quantum criticality~\cite{mukherjee2021many,murphy2025ergotropy}. More broadly, our results suggest that ancilla-assisted quantum batteries provide a useful platform for correlation-enhanced energy extraction, thermodynamic sensing of spectral transitions, and the design of robust quantum energy-storage architectures in noisy quantum technologies~\cite{campbell2026roadmap}.
	
	\appendix
	
	\section{Analytical evaluation of ergotropy and daemonic gain for unitary evolution}\label{appa}

	In this appendix we present the analytical evaluation of the ergotropy and the corresponding daemonic gain for the non-interacting qubit system considered in Sec.~\ref{sec3}. The system consists of two battery qubits ($S$) and a single auxiliary ancilla qubit ($A$). The initial state is prepared as a correlated superposition between the ground and first excited states of the system Hamiltonian $H_S$,
	\begin{equation}
		|\Psi(0)\rangle
		=
		\alpha |11\rangle_S |0\rangle_A
		+
		\beta |01\rangle_S |1\rangle_A ,
	\end{equation}
	where $\alpha^2+\beta^2=1$. Here $|1\rangle$ and $|0\rangle$ denote the excited and ground states of each qubit, respectively.
	
	The total Hamiltonian governing the charging dynamics is
	\begin{equation}
		H = H_S + H_C .
	\end{equation}
	Since the ancilla does not participate in the charging process, its Hamiltonian can be ignored and effectively acts as the identity operator $\mathbb{I}_A$. Furthermore, the charging field is applied locally and identically to the system qubits, allowing the total Hamiltonian to be decomposed into independent single-qubit contributions.
	
	For each qubit the local Hamiltonian can be written as
	\begin{equation}
		H^{(i)} =\frac{\Omega}{2} \left(\frac{h}{\Omega}\sigma_z + \frac{2 \omega_c}{\Omega}\sigma_x\right),
	\end{equation}
	with
	\begin{equation}
		\Omega = \sqrt{h^2 + 4\omega_c^2}.
	\end{equation}	
	The single-qubit evolution operator is therefore
	\begin{equation}
		U^{(i)}(t)=e^{-iH^{(i)}t}=e^{-i\frac{\Omega t}{2} [\hat{a}.\vec{\sigma}]},
	\end{equation}
	With $\hat{a}=(2\omega_c/\Omega )\hat{x}+(h/\Omega) \hat{z}$ and $\vec{\sigma}$ is the Pauli matrix vector. Because the qubits evolve independently, the global evolution operator acting on the system is
	\begin{equation}
		U(t)=U^{(1)}(t)\otimes U^{(2)}(t).
	\end{equation}
	
	Applying $U(t)$ to the system components of the initial state produces the evolved states
	\begin{align}
		|\phi_0(t)\rangle &= U(t)|11\rangle
		= (U^{(1)}(t)|1\rangle)\otimes(U^{(2)}(t)|1\rangle), \\
		|\phi_1(t)\rangle &= U(t)|01\rangle
		= (U^{(1)}(t)|0\rangle)\otimes(U^{(2)}(t)|1\rangle).
	\end{align}
	The instantaneous energy of the system with respect to the system Hamiltonian
	\begin{equation}
		H_S=\frac{h}{2}\left(\sigma_z^{(1)}+\sigma_z^{(2)}\right)
	\end{equation}
	is given by
	\begin{equation}
		E(t)=\alpha^2\langle\phi_0(t)|H_S|\phi_0(t)\rangle
		+\beta^2\langle\phi_1(t)|H_S|\phi_1(t)\rangle .
	\end{equation}
	The ergotropy is obtained by subtracting the passive energy $E_P$, defined as the minimal energy obtainable through unitary rearrangement of the eigenvalues of the density matrix. This yields
	\begin{equation}
		\xi(t)
		=
		\frac{8\alpha^2 h \omega_c^2}{\Omega^2}
		\sin^2\!\left(\frac{\Omega t}{2}\right)
		-\alpha^2 h
		+
		h\,\max(\alpha^2,\beta^2),
	\end{equation}
	where the $\max$ function arises from ordering the eigenvalues in increasing energy to construct the passive state.
	
	For the unitary case, the daemonic gain can be obtained directly using the result of Theorem~1, according to which globally pure system–ancilla states saturate the bound energy limit. Consequently, the daemonic gain is equal to the bound energy of the system state and evaluates to
	\begin{equation}
		\delta W = h\,\min(\alpha^2,\beta^2).
		\label{appeq}
	\end{equation}
	
	For the interacting Hamiltonian introduced in Eq.~\eqref{2e1}, an analogous argument yields the expression for daemonic gain under unitary dynamics as
	\begin{equation}
		\delta W = |e_2-e_4|\min(\alpha^2,\beta^2),
	\end{equation}
	where $e_2$ and $e_4$ denote the relevant lower energy eigenvalues of the interacting system Hamiltonian. This expression shows that the attainable daemonic gain is directly determined by the spectral separation of the low-energy states and the initial system–ancilla correlations. Substituting the values of the energy levels,
	\begin{equation}
		\delta W = |-\sqrt{J^2 + D^2} - J_z+\sqrt{h^2 + J^2 \gamma^2}|\min(\alpha^2,\beta^2).
	\end{equation}

	\bibliography{c2}
	
\end{document}